\documentclass[10pt,a4paper]{article}

\usepackage{times}
\usepackage{graphicx}
\usepackage{amsmath,amssymb,amsthm}
\usepackage{algorithm}
\usepackage{algorithmic}
\usepackage{tikz}
\usetikzlibrary{positioning,calc}
\usepackage{xcolor}
\usepackage[square,authoryear]{natbib}
\usepackage[capitalize,noabbrev]{cleveref}
\usepackage{authblk}
\usepackage{a4wide}

\usepackage{algorithm}
\usepackage{algorithmic}
\usepackage[T1]{fontenc}

\newcommand{\PROCEDURE}[2]{\STATE \textbf{function} #1(#2) \begin{ALC@g}}
\newcommand{\ENDPROCEDURE}{\end{ALC@g} \STATE \textbf{end function}}
\newcommand{\SWITCH}[1]{\STATE \textbf{switch} (#1) \begin{ALC@g}}
\newcommand{\ENDSWITCH}{\end{ALC@g} \STATE \textbf{end switch}}
\newcommand{\CASE}[1]{\STATE \textbf{case} #1\textbf{:} \begin{ALC@g}}
\newcommand{\ENDCASE}{\end{ALC@g}}

\newcommand{\DEFAULT}{\STATE \textbf{default:} \begin{ALC@g}}
\newcommand{\ENDDEFAULT}{\end{ALC@g}}
\newcommand{\DEFAULTLINE}[1]{\STATE \textbf{default:} }

\newcommand{\pluseq}{\mathrel{+}=}

\newtheorem{theorem}{Theorem}
\newtheorem{lemma}{Lemma}
\newtheorem{corollary}{Corollary}
\newtheorem{example}{Example}
\newtheorem{proposition}{Proposition}
\newtheorem{definition}{Definition}

\begin{document}

\title{Group Vitality Indices: Axioms and Algorithms}

\author{Natalia Kucharczuk}
\author{Oskar Skibski}

\affil{\normalsize University of Warsaw, Banacha 2, 02-097 Warszawa, Poland}
\affil{\textit{kucharczuk.natalia@gmail.com},
\textit{oskar.skibski@mimuw.edu.pl}}

\maketitle

\begin{abstract}
We consider the problem of assessing a group of nodes in a network.
Our focus is on vitality indices---a natural class of centrality measures that evaluate the importance of a node by examining the impact of its removal on the network.
We conduct a comprehensive analysis of group vitality indices.
Specifically, we show that every vitality index admits a unique extension to groups, which can be defined using a group variant of the Shapley value recently proposed in the literature.
We also provide an axiomatization of the entire class, along with two specific group vitality indices that satisfy additional normalization conditions.
Furthermore, we study the computational properties of all vitality indices, as well as Group Attachment Centrality.
\end{abstract}

\section{Introduction}

The literature on centrality measures---a key tool in social network analysis---predominantly concentrates on individual nodes~\citep{Brandes:Erlebach:2005}.
However, in practice, centrality is often used to select a \emph{group} of nodes that collectively serve a common goal~\citep{Kumar:etal:2019,Wan:etal:2021,Suri:Narahari:2008}.
Selecting nodes independently, one by one, ignores the synergies and interdependencies among them and leads to suboptimal solutions.

\emph{Group centrality measures} address this limitation by providing a more holistic assessment of groups of elements~\citep{Camur:Vogiatzis:2024,Kolaczyk:etal:2009}.
However, extending node centralities to groups is a nontrivial challenge.
For many centrality measures, it is unclear what the appropriate group-level generalization should be.
Two simple ideas are (a) to sum the centralities of individual nodes, and (b) to merge the group into a single node and compute its centrality in the resulting graph.
Both approaches have clear shortcomings: the sum approach ignores the co-dependencies among nodes, while the merge approach significantly alters the graph's structure.
See Figure~\ref{fig:two_paths} for an example.
For this reason, alternative extensions have been proposed for many centrality measures.
For example, for distance-based centralities---such as closeness centrality---several variants of the group definitions have been proposed based on the distance from the group to the closest, farthest, or average node \citep{Everett:Borgatti:1999}.

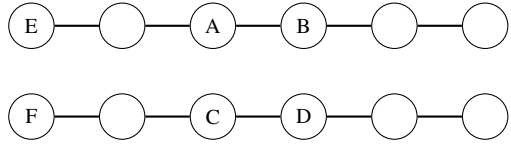
\begin{figure}[t]
\centering
\begin{tikzpicture}[x=6cm,y=6cm] 
  \tikzset{     
    e4c node/.style={circle,draw,minimum size=0.6cm,inner sep=0,font=\footnotesize}, 
    e4c edge/.style={sloped,above,font=\footnotesize}
  }
  \node[e4c node] (1) at (0.0, 0.2) {E}; 
  \node[e4c node] (2) at (0.2, 0.2) {}; 
  \node[e4c node] (3) at (0.4, 0.2) {A}; 
  \node[e4c node] (4) at (0.6, 0.2) {B}; 
  \node[e4c node] (5) at (0.8, 0.2) {}; 
  \node[e4c node] (6) at (1.0, 0.2) {}; 
  \node[e4c node] (7) at (0.0, 0.00) {F}; 
  \node[e4c node] (8) at (0.2, 0.00) {}; 
  \node[e4c node] (9) at (0.4, 0.00) {C}; 
  \node[e4c node] (10) at (0.6, 0.00) {D}; 
  \node[e4c node] (11) at (0.8, 0.00) {}; 
  \node[e4c node] (12) at (1.0, 0.00) {}; 

  \path[draw,thick]
  (1) edge[e4c edge]  (2)
  (2) edge[e4c edge]  (3)
  (3) edge[e4c edge]  (4)
  (4) edge[e4c edge]  (5)
  (5) edge[e4c edge]  (6)
  (7) edge[e4c edge]  (8)
  (8) edge[e4c edge]  (9)
  (9) edge[e4c edge]  (10)
  (10) edge[e4c edge]  (11)
  (11) edge[e4c edge]  (12)
  ;
\end{tikzpicture}
\caption{Consider the problem of selecting two nodes in the graph above, based on Flow Betweenness Centrality. The sum approach (which evaluates a group by the sum of its members' centralities) does not distinguish between the groups $\{A,B\}$ and $\{A,C\}$. In contrast, the merge approach (which evaluates a group by the centrality of the merged entity) assigns a high value to the group $\{E,F\}$---two peripheral and otherwise irrelevant nodes---since, after their artificial merger, the structure of the graph is substantially altered and they become positioned at the center of an extended path.}\label{fig:two_paths}
\end{figure}

\emph{Vitality indices}~\citep{Brandes:Erlebach:2005}, also referred to as \emph{induced centrality measures}~\citep{Everett:Borgatti:2010}, form an important class of centrality measures.
They are based on a natural idea: assessing a node's importance by evaluating the impact of its removal on the network.
As a result, they are well-suited to identifying the elements whose removal causes the greatest disruption to the network.
This is particularly relevant in robustness analysis of computer networks or the WWW~\citep{Albert:etal:2000} (e.g., designing fault-tolerant architectures, protecting vital nodes against cyber-attacks), but also transportation networks~\citep{Rosato:etal:2008} (e.g., designing redundant infrastructure) and social networks~\citep{Ressler:2006} (e.g., epidemic control, target prioritization in terrorist networks).
Several classical centralities---such as Degree Centrality and Flow Betweenness---fall within this category.
Moreover, this class includes many so-called \emph{game-theoretic centrality measures}, which are defined via coalitional games and often rely on the Shapley value~\citep{Lindelauf:etal:2013,Skibski:etal:2018:gtc}.

Against this background, in this paper, we conduct the very first comprehensive analysis of group vitality indices and address several fundamental questions:

\begin{itemize}
\item \emph{Existence and Uniqueness}:
We show that every vitality index admits a \emph{unique} extension to groups.
Specifically, we define a class of group vitality indices and prove that each vitality index corresponds to a unique group variant within this class.
To the best of our knowledge, this is the first general existence and uniqueness result for any broad class of centrality measures.

\item \emph{Axiomatic Characterization}:
We provide two axiomatic characterizations of group vitality indices, based on \emph{Group Balanced Contribution} and \emph{Vertex Exclusion}.
Our results also reveal a direct connection to a group version of the Shapley value.
In addition, we axiomatically characterize two specific measures---\emph{Group Degree} and \emph{Group Attachment Centrality}---which satisfy dual normalization conditions.
The axiomatic approach is fundamental in centrality analysis, as it highlights similarities and differences among the many measures proposed in the literature.

\item \emph{Computational Complexity}:
We show that, under mild assumptions on the underlying vitality index, finding a group that maximizes centrality is NP-hard.
However, computing the centrality of a given group is no harder than computing individual node centralities; more precisely, it is \emph{polynomial-time computable} whenever the node-level vitality index is polynomial-time computable.
We also develop an approximation algorithm for Group Attachment Centrality based on a novel sampling method for the group Shapley value.
This algorithm also supports efficient computation of node-level Attachment Centrality.
\end{itemize}

\paragraph{Related Work}
Our theoretical results contribute to an active line of research on the axiomatic foundations of centrality measures. 
This line was initiated by \citet{Sabidussi:1966} and \citet{Nieminen:1973} over 50 years ago, and is currently being studied mostly in the AI, electronic commerce, and social choice communities.

Notably, \citet{Altman:Tennenholtz:2005} proposed an axiomatization of the Seeley Index, also known as the \emph{simplified PageRank}, and \citet{Was:Skibski:2023:pagerank} developed an axiomatic characterization of PageRank in its full form.
In turn, \citet{Garg:2009} studied distance-based centralities, and \citet{Was:etal:2019:rwd} focused on random-walk-based ones.
More recently, in a more general approach, \citet{Bloch:etal:2023} developed a joint axiomatic framework for centrality measures.
Also, \citet{Boldi:Vigna:2014} analyzed the satisfiability of several axioms and concluded that Harmonic Centrality is the only popular centrality measure that satisfies all of them.

Most closely related are works that analyze vitality indices~\citep{Everett:Borgatti:2010,Skibski:2021:vitality} and game-theoretic centralities~\citep{Michalak:etal:2013:connectivity,Skibski:etal:2018:gtc}.
However, none of these studies focus on group centralities.
In fact, our work is the first to address the axiomatic characterization of group centrality measures.
We will discuss the most relevant results on vitality indices in the Preliminaries section.

Our algorithmic results are loosely related to research on the Shapley value in graph settings~\citep{Greco:etal:2020}.

\section{Preliminaries}

A graph is a pair $G = (V,E)$ such that $V$ is the set of nodes and $E$ is the set of edges, i.e., unordered pairs of nodes $\{u,v\} \subseteq V$.  
The set of all possible graphs is denoted by $\mathcal{G}$.

The \emph{degree} of a node is the number of its incident edges: $\deg(v) = |\{ e \in E : v \in e \}|$.  
A node $v$ is \emph{isolated} if $\deg(v) = 0$.  
We refer to nodes with degree one as \emph{leaf nodes}.

A graph is \emph{connected} if there exists a path between any two nodes.  
For a given set of nodes $S \subseteq V$, the subgraph induced by $S$ is the graph $G[S] = (S,E[S])$, where $E[S]$ is the set of all edges from $E$ between nodes in $S$.  
The subgraphs of $G$ induced by the nodes $V \setminus \{v\}$ and $V \setminus S$ are denoted by $G - v$ and $G - S$, respectively.

A maximal subset of nodes that induces a connected subgraph is called a \emph{component}.  
The set of all components of a graph $G$ is denoted by $K(G)$.

A (node) centrality measure $F$ is a function that, given a graph $G = (V,E)$ and a node $v \in V$, returns a real value $F_v(G)$ that corresponds to the node's importance---the higher the value, the more important the node.  
A \emph{group centrality measure} $F'$ assesses groups of nodes: given a graph $G = (V,E)$ and a group of nodes $S \subseteq V$, it returns a real value $F_S(G)$.  
We assume $F_{\emptyset}(G) = 0$.  

A group centrality measure $F'$ \emph{extends} a node centrality $F$ if $F'_{\{v\}}(G) = F_v(G)$ for every graph $G = (V,E)$ and node $v \in V$.  
Two simple extensions of a node centrality to groups are:
\begin{itemize}
\item \emph{sum approach}: $F'_S(G) = \sum_{v \in S} F_v(G)$.
\item \emph{merge approach}: $F'_S(G) = F_{[S]}(G_{[S]})$, where graph $G_{[S]}$ is obtained from $G$ by merging nodes from $S$ into one node: $[S]$.
\end{itemize}
Clearly, the sum approach does not provide any additional information.  
In turn, the merge approach is not always desirable, as it modifies the whole structure of the graph.

\subsection{Coalitional games}
Some of our results connect centrality analysis with cooperative game theory~\citep{Chalkiadakis:etal:2011}.  
A \emph{coalitional game} is a pair $(N,g)$ where $N$ is the set of players and $g \colon 2^N \rightarrow \mathbb{R}$ is a characteristic function that evaluates each \emph{coalition}, i.e., a group of players (with the assumption that $g(\emptyset) = 0$).  
A classic measure of the importance of individual players in a game is the \emph{Shapley value}: in a game $(N,g)$, the importance of a player $i \in N$ is given by:
\[
SV_i(N,g) = \sum_{\substack{T \subseteq N\\ i \in T}} \frac{(|T|\!-\!1)!(|N|\!-\!|T|)!}{|N|!} \big(g(T) - g(T \setminus \{i\})\big).
\]
This expression defines the value of a player as a weighted average of their marginal contributions across all possible coalitions.
The importance of the Shapley value stems from its axiomatic characterization. Specifically, \citet{Shapley:1953} proved that it is the unique method that satisfies several desirable axioms: \emph{Efficiency}, \emph{Symmetry}, \emph{Additivity}, and \emph{Null-player}.
More recently, the Shapley value has emerged as a key tool in explainable artificial intelligence, providing principled methods for attributing predictions to input features~\citep{Lundberg:Lee:2017,Fryer:etal:2021}.

The Shapley value has been also extended to groups.
One of the classic method was proposed by \citet{Marichal:etal:2007}:
\[ SV_S(N,g) \!=\! \sum_{T \subseteq N \setminus S} \frac{|T|! (|N|\!-\!|T|\!-\!|S|)!}{(|N|-|S|+1)!} (g(T \cup S) - g(T)). \]
This method, roughly speaking, corresponds to the merge approach in the graph context: it is the payoff of a player obtained by merging coalition $S$ into one player.
Alternative definition was recently proposed by \citep{Kepczynski:Skibski:2026:union} under the name \emph{Union Shapley value}:
\[
US_S(N,g) = \sum_{T \subseteq N} \frac{(|T| - 1)! (|N| - |T|)!}{|N|!} (g(T) - g(T \setminus S)).
\]
Clearly, for both definitions, if $S = \{i\}$, the group Shapley value coincides with the single-player definition.
Other group definitions include the Interaction Index \citep{Grabisch:Roubens:1999}.

\subsection{Vitality Indices}

A centrality measure $F$ is a \emph{vitality index} if there exists an invariant function $f \colon \mathcal{G} \rightarrow \mathbb{R}$ such that $F_v(G) = f(G) - f(G - v)$.  
Two classic vitality indices are:
\begin{description}
\item[Degree Centrality:] $D_v(G) = \deg(v)$ is a vitality index with the invariant function $f(G) = |E|$ for every $G = (V,E)$.
\item[Flow Betweenness Centrality \citep{Freeman:etal:1991}:]\ \\ $FB_v(G) = \sum_{s,t \in V} \big(\text{flow}_{s,t}(G) - \text{flow}_{s,t}(G-v)\big)$, where $\text{flow}_{s,t}(G)$ is the number of edge-disjoint paths from $s$ to $t$ in $G$.
\end{description}

Another group of vitality indices is formed by \emph{game-theoretic centralities}.  
These are measures defined via coalitional games, almost uniformly using the Shapley value.  
The idea is to define a coalitional game based on a graph---where nodes are players---and then compute the Shapley value to obtain the importance of each node in the game, and thus in the graph.  
Examples include:

\begin{description}
\item[Connectivity Centrality~\citep{Amer:Gimenez:2004}:] $CN_v(G) = SV_v(V,g^{CN})$, where $g^{CN}(T) = 1$ if $G[T]$ is connected, and $g^{CN}(T) = 0$ otherwise, for every $T \subseteq V$.
\item[Attachment Centrality~\citep{Skibski:etal:2019:attachment}:] $A_v(G) = SV_v(V,g^A)$, where $g^A(T) = 2(|T| - |K(G[T])|)$ for every $T \subseteq V$.
\end{description}

Both of these game-theoretic centralities are \emph{induced}, meaning that the value of a coalition $T$ in the game depends solely on the subgraph $G[T]$.  
\citet{Skibski:2021:vitality} proved that vitality indices are equivalent to induced game-theoretic centralities.  

Formally, every vitality index $F$ for every graph $G = (V,E)$ and $v \in V$ satisfies 
\[ F_v(G) = SV_v(V,g^*), \mbox{ where } g^*(T) = \sum_{v \in T} F_v(G[T]).\] 
Hence, it is an induced game-theoretic centrality.  
Conversely, for every induced game-theoretic centrality, there exists a function $f$ such that $F_v(G) = f(G) - f(G - v)$, which implies that it is a vitality index.

This result follows from the axiomatic characterization of vitality indices~\citep{Skibski:2021:vitality}.  
Specifically, a centrality measure is a vitality index if and only if it satisfies \emph{Balanced Contribution}---a property defined by \citet{Myerson:1980} in game-theoretic literature:
\begin{description}
\item[Balanced Contributions:] For every graph $G = (V,E)$ and two nodes $u,v \in V$, it holds that:
\[
F_v(G) - F_v(G - u) = F_u(G) - F_u(G - v).
\]
\end{description}
In particular, as a result of the Shapley value's axiomatic characterization, every vitality index satisfies the axiom of \emph{Fairness}.
This property in the network context means that adding an edge increases the centrality of both its endpoints equally.

\section{Group Vitality Indices}

In this section, we introduce the group variant of vitality indices and present several important results concerning it.

Group vitality indices have a natural definition.  
Let $F$ be an arbitrary group centrality measure.  
We say it is a \emph{group vitality index} if there exists an invariant function $f: \mathcal{G} \rightarrow \mathbb{R}$ such that $F_S(G) = f(G) - f(G - S)$.

We begin with a simple but important observation.  
If $F$ is a (node) vitality index with an invariant function $f$, then a group vitality index $F'$ with the same function $f$ naturally extends $F$.
Clearly, this is not the only group centrality that extends $F$, as centralities obtained through the sum and merge approach also extend it.
However, as we now show, it is the only \emph{group vitality index} that extends $F$.

\begin{theorem}\label{theorem:one_extends}
For a vitality index, there exists exactly one group vitality index that extends it.
\end{theorem}

\begin{proof}
Fix a vitality index $F$.  
Assume there are two centrality indices $F'$ and $F''$ with invariant functions $f'$ and $f''$, respectively, that extend $F$.  
Let $G$ be the graph with the smallest number of vertices where $F'$ and $F''$ differ, and let $S$ be the group on which they differ.  
Group $S$ has at least two nodes, as both centralities agree on single nodes and assign zero centrality to the empty group.  
Take $v \in S$. We get that:
\begin{align*}
F'_S(G) &= f'(G) - f'(G - S) \\
&= f'(G) - f'(G - v) + f'(G - v) - f'(G - S) \\
&= F'_{\{v\}}(G) + F'_{S \setminus \{v\}}(G - v).
\end{align*}
Similarly, $F''_S(G) = F''_{\{v\}}(G) + F''_{S \setminus \{v\}}(G - v)$.  
However, graph $G - v$ has fewer vertices than $G$, which implies $F'_{S \setminus \{v\}}(G - v) = F''_{S \setminus \{v\}}(G - v)$, which combined with the fact that $F'_{\{v\}}(G) = F''_{\{v\}}(G)$ leads to a contradiction.
\end{proof}

To give an example, consider Degree Centrality.  
We have multiple options to extend it to groups:
\begin{itemize}
\item using the sum approach we can sum degrees of all its members: $F_S(G) = \sum_{v \in S} \deg(v)$; in this way, we count all edges between $S$ and $V \setminus S$ once, and edges between nodes from $S$ twice;
\item using the merge approach, we can count the neighbours of the group $S$, ignoring elements from $S$; thus, we do not count edges between members of $S$ and count each edge to a node $v$ from $V \setminus S$ only once.
\end{itemize}

In turn, Theorem~\ref{theorem:one_extends} states that if we look at Degree Centrality from the perspective of a vitality index, then there is only one group vitality index that extends it.

\begin{definition}
\emph{Group (Vitality) Degree Centrality} is defined as
\begin{equation}\label{eq:group_degree}
D_S(G) = |E| - |E[V \setminus S]|
\end{equation}
for every graph $G = (V,E)$ and every group $S \subseteq V$.
\end{definition}
Thus, we count each edge incident to nodes from $S$ exactly once.
See Figure~\ref{fig:degree} for an illustration.

Consider now another vitality index, \emph{Attachment Centrality}, defined as the Shapley value of the cooperative game $g^{A}(T) = 2\bigl(|T| - |K(G[T])|\bigr)$.
It can be extended to groups using both standard approaches, as well as by replacing the Shapley value with one of the existing group variants of the Shapley value.
However, Theorem~\ref{theorem:one_extends} indicates that there is only one extension that qualifies as a group vitality index.
To identify it, we rely on the invariant function.

\begin{definition}
\emph{Group Attachment Centrality} is defined as
\begin{equation*}
A_S(G) = f^{A}(G) - f^{A}(G - S),
\end{equation*}
where $f^{A}(G)$ is the invariant function of (node-level) Attachment Centrality, defined for every graph $G = (V,E)$ as
\begin{equation*}
f^{A}(G)
= \sum_{S \subseteq V}
\frac{(|S| - 1)! (|V| - |S|)!}{|V|!}
\cdot 2\bigl(|S| - |K(G[S])|\bigr).
\end{equation*}
\end{definition}
By reformulation, $f^{A}$ can be interpreted as a network-level metric that captures how well connected the graph is:
\[
f^{A}(G)
= \frac{1}{|V|}
\sum_{s = 1}^{|V|}
\frac{1}{\binom{|V|}{s}}
\sum_{\substack{S \subseteq V \\ |S| = s}}
\frac{2|V|}{s}\,\bigl(s - |K(G[S])|\bigr).
\]
Hence, this function computes an average over all subgraph sizes.
For each size $s$, it averages over all induced subgraphs of size $s$, and for each such subgraph evaluates a quantity that depends on the number of connected components relative to its size.
This value ranges from $0$, when all nodes in the subgraph are isolated, to $2|V|(s-1)/s$, when the subgraph is fully connected.

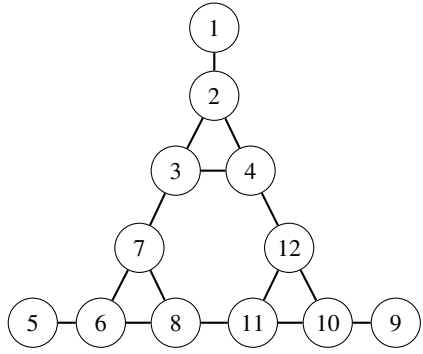
\begin{figure}[tb!]
\centering
\begin{tikzpicture}[x=3cm,y=3cm]
  \tikzset{     
    e4c node/.style={circle,draw,minimum size=0.65cm,inner sep=0,font=\small}, 
    e4c edge/.style={sloped,above,font=\footnotesize}
  }
  \def\r{1}                 
  \def\h{0.866}             

  \node[e4c node] (1) at (0.5, 1.30) {1}; 
  \node[e4c node] (2) at (0.5, 1.00) {2}; 
  \node[e4c node] (3) at (0.33, 0.67) {3}; 
  \node[e4c node] (4) at (0.66, 0.67) {4}; 
  \node[e4c node] (5) at (-0.30, 0.00) {5}; 
  \node[e4c node] (6) at (0.00, 0.00) {6}; 
  \node[e4c node] (7) at (0.17, 0.33) {7}; 
  \node[e4c node] (8) at (0.33, 0.0) {8}; 
  \node[e4c node] (9) at (1.30, 0.00) {9}; 
  \node[e4c node] (10) at (1.00, 0.00) {10}; 
  \node[e4c node] (11) at (0.67, 0.00) {11}; 
  \node[e4c node] (12) at (0.83, 0.33) {12}; 

  \path[draw,thick]
  (1) edge[e4c edge]  (2)
  (2) edge[e4c edge]  (3)
  (2) edge[e4c edge]  (4)
  (3) edge[e4c edge]  (4)
  (5) edge[e4c edge]  (6)
  (6) edge[e4c edge]  (7)
  (6) edge[e4c edge]  (8)
  (7) edge[e4c edge]  (8)
  (9) edge[e4c edge]  (10)
  (10) edge[e4c edge]  (11)
  (10) edge[e4c edge]  (12)
  (11) edge[e4c edge]  (12)
  (3) edge[e4c edge]  (7)
  (8) edge[e4c edge]  (11)
  (12) edge[e4c edge]  (4)
  ;
  
\end{tikzpicture}
\caption{Consider three groups: $A = \{2,3,4\}$, $B = \{4,7,11\}$, and $C = \{2,6,10\}$.  
Degree Centrality extended using the sum approach does not differentiate between these groups.  
When extended using the merge approach, common neighbors are avoided, making group $C$ the best.  
In contrast, Group Degree Centrality (see Equation~\eqref{eq:group_degree}) only avoids shared edges; therefore, groups $B$ and $C$ are better than $A$, but equally good.  
Group Attachment Centrality identifies group $C$ as the best group of size 3 in the graph, as its removal splits the graph into four components.}
\label{fig:degree}
\end{figure}

To axiomatically characterize the class of group vitality indices, we introduce two axioms.  
The first, \emph{Group Balanced Contribution}, is a natural extension of Balanced Contributions~\citep{Myerson:1980} to groups.
It states that the removal of a group $S$ affects a group $T$ in the same way that the removal of $T$ affects $S$.  
This axiom can be interpreted as a stability condition, familiar from the Shapley value axiomatization: if the group assessment represents a payoff or award, the ``threat'' posed by one group on the other is \emph{balanced} by the reciprocal threat.  
Moreover, it implies \emph{Fairness}, which requires that adding an edge affects both of its endpoints in the same way.

\begin{description}
\item[Group Balanced Contributions:] For every graph $G = (V,E)$ and two groups $S,T \subseteq V$, it holds that
\[
F_S(G) - F_{S \setminus T}(G - T) = F_T(G) - F_{T \setminus S}(G - S).
\]
\end{description}

The second axiom, \emph{Vertex Exclusion}, is novel and states that the marginal contribution of a node to a group in a graph is equal to its centrality in that graph.  
This axiom provides a natural approach for extending node-level measures to groups, distinct from the standard sum and merge approaches: the value of a group can be computed by selecting an arbitrary node, adding its centrality in the current graph, and then removing it; the process is repeated with the remaining nodes until the group is exhausted.  
Interestingly, while this procedure may generally depend on the order of node selection, it is order-independent in the case of vitality indices.

\begin{description}
\item[Vertex Exclusion:] For every graph $G = (V,E)$, group $S \subseteq V$, and $v \in S$, it holds that
\[
F_S(G) = F_{\{v\}}(G) + F_{S \setminus \{v\}}(G - v).
\]
\end{description}

The following theorem states that each of these axioms uniquely characterizes the whole class.

\begin{theorem}\label{theorem:axioms}
Given a group centrality measure $F$, the following three statements are equivalent:
\begin{itemize}
\item[(a)] $F$ is a group vitality index;
\item[(b)] $F$ satisfies Vertex Exclusion;
\item[(c)] $F$ satisfies Group Balanced Contributions.
\end{itemize}
\end{theorem}
\begin{proof}[Proof (sketch)]
It is immediate that $(a) \Rightarrow (b)$. 

To prove $(b) \Rightarrow (c)$, observe that
\[ F_T(G) = \sum_{i=1}^{t} F_{v_i}\bigl(G - \{v_1,\dots,v_{i-1}\}\bigr), \]
which implies
\[ F_{S \cup T}(G) = F_T(G) + F_{S \setminus T}(G - S), \]
by applying Vertex Exclusion $|T|$ times.
By symmetry of $S$ and $T$, we obtain Group Balanced Contribution.

Finally, the proof of $(c) \Rightarrow (a)$ relies on the Shapley value-like form of measures that satisfy Balanced Contributions.
\begin{lemma}\label{lemma:union_shapley_form}
If a group vitality index $F$ satisfies Group Balanced Contributions, then it satisfies:
\[
F_S(G) = \sum_{T \subseteq V} \frac{(|T| - 1)! (|V| - |T|)!}{|V|!} (g^*(T) - g^*(T \setminus S)).
\]
for every graph $G = (V,E)$ and $S \subseteq V$, where $g^*(T) = \sum_{v \in T} F_v(G[T])$ for every $T \subseteq V$.
\end{lemma}
The full proof can be found in the appendix.
\end{proof}

Lemma~\ref{lemma:union_shapley_form} reveals a connection to coalitional games.  
Specifically, out of several extensions of the Shapley value to groups, only one---the Union Shapley value~\citep{Kepczynski:Skibski:2026:union}---adheres to the definition of group vitality indices.  
From Lemma~\ref{lemma:union_shapley_form}, we obtain the following result, analogous to the result for (node) vitality indices.

\begin{corollary}\label{corollary:union_shapley}
For every group vitality index $F$, every graph $G = (V,E)$ and coalition $S \subseteq V$ it holds:
\[ F_S(G) = US_S(V,g^*), \mbox{ where }g^*(T) = \sum_{v \in T} F_v(G[T]). \]
\end{corollary}

\begin{proof}
Immediate from Theorem~\ref{theorem:axioms} and Lemma~\ref{lemma:union_shapley_form}.
\end{proof}

\section{Group Degree and Attachment: Two Normalized Measures}

In this section, we will provide axiomatizations of two dual group vitality indices.
We begin by presenting two natural axioms that concern nodes with at most one edge.

\begin{description}
\item[Isolated Node] For every graph $G = (V,E)$ and an isolated node $v \in V$ it holds $F_v(G) = 0$.
\item[Leaf Node] For every graphs $G = (V,E)$, $G'=(V',E')$ and nodes $v \in V$, $v' \in V$ with $\deg(v) = 1 = \deg(v')$ it holds $F_v(G) = F_{v'}(G')$.
\end{description}

If a node is isolated, it clearly has zero importance; the Isolated Node axiom states its centrality is zero which is consistent with the vast majority of centrality measures.
Furthermore, if a node has only one edge, then it has a minimal impact on the structure of the network.
The Leaf Node axiom states such nodes always have the same centrality.
For example, for Betweenness Centrality we have $c = 0$, and for Degree Centrality: $c = 1$.
However, for some measures, e.g., for PageRank, the centrality of such nodes depend on the position in the network of their neighbours.

Now, consider an arbitrary connected component of a graph.
What should be the sum of centralities of its elements?
One approach would be to say the sum of centralities in each component should be normalized so it does depend only on the size of the component. 
In turn, an alternative approach would be to say the sum of centralities in each component should depend not on the number of nodes, but on the number of edges.
These two ideas lead to the following axioms.

\begin{description}
\item[Node-Normalization] There exists a function $h: \mathbb{N} \rightarrow \mathbb{R}$ such that for every graph $G = (V,E)$ and every connected component $C \in K(G)$ it holds: $\sum_{v \in C} F_v(G) = h(|C|)$.
\item[Edge-Normalization] There exists a function $h: \mathbb{N} \rightarrow \mathbb{R}$ such that for every graph $G = (V,E)$ and every connected component $C \in K(G)$ it holds: $\sum_{v \in C} F_v(G) = h(|E[C]|)$.
\end{description}

We will now show that for each type of normalization axiom there exists exactly one centrality measure (up to a scalar multiplication).

\begin{theorem}\label{theorem:axioms_attachment}
A group vitality index satisfies Isolated Node, Leaf Node and Node-Normalization if and only if it is Group Attachment Centrality (up to a scalar multiplication).
\end{theorem}

\begin{theorem}\label{theorem:axioms_degree}
A group vitality index satisfies Isolated Node, Leaf Node and Edge-Normalization if and only if it is Group Degree Centrality (up to a scalar multiplication).
\end{theorem}

\begin{proof}[Proof (sketch)]
The proof is based on the following steps:
\begin{itemize}
\item First, it is shown that for every vitality index satisfying Isolated Node and Leaf Node it holds $F_v(G^{\star}) = c \cdot D_v(G^{\star})$ for any graph $G^{\star} = (V,E)$ and node $v \in V$ and some constant $c \in \mathbb{R}$.
\item Using Node- and Edge-Normalization, we get that the sum of centralities in every connected component is fixed and equal to the case of a star graphs.
\item Now, Lemma~\ref{lemma:union_shapley_form} implies uniqueness of the measure (up to a scalar multiplication).
\end{itemize}
The full proof can be found in the appendix.
\end{proof}

\begin{figure}[t!]
\centering
\begin{tikzpicture}[x=6.5cm,y=6.5cm,rotate=90] 
  \tikzset{     
    e4c node/.style={circle,draw,minimum size=0.5cm,inner sep=0,text=black}, 
    e4c edge/.style={sloped,above,font=\footnotesize}
  }
\node[e4c node,fill=blue!31] (15) at (0.32, 1.20) {15}; 
\node[e4c node,fill=blue!20] (2) at (0.64, 0.62) {2}; 
\node[e4c node,fill=blue!31.8] (4) at (0.48, 1.05) {4}; 
\node[e4c node,fill=blue!26.6] (5) at (0.39, 0.61) {5}; 
\node[e4c node,fill=blue!46.6] (7) at (0.89, 1.03) {7}; 
\node[e4c node,fill=blue!26.6] (8) at (0.82, 1.22) {8}; 
\node[e4c node,fill=blue!26.6] (9) at (1.00, 1.19) {9}; 
\node[e4c node,fill=blue!46.6] (10) at (0.83, 0.86) {10}; 
\node[e4c node,fill=blue!60] (1) at (0.53, 0.73) {1}; 
\node[e4c node,fill=blue!40] (6) at (0.65, 0.84) {6}; 
\node[e4c node,fill=blue!34.2] (13) at (0.00, 1.15) {13}; 
\node[e4c node,fill=blue!41.6] (16) at (0.51, 0.91) {16}; 
\node[e4c node,fill=blue!30] (17) at (0.06, 0.82) {17}; 
\node[e4c node,fill=blue!34.2] (18) at (0.34, 1.02) {18}; 
\node[e4c node,fill=blue!35] (19) at (0.19, 0.87) {19}; 
\node[e4c node,fill=blue!31.2] (12) at (0.58, 1.21) {12}; 
\node[e4c node,fill=blue!51.2] (11) at (0.65, 1.02) {11}; 
\node[e4c node,fill=blue!65.8] (14) at (0.16, 1.09) {14}; 
\node[e4c node,fill=blue!54.6] (3) at (0.30, 0.77) {3}; 

  \node[right=1pt of 1] {Hani Hanjour};
  \node[below right=-2pt of 3] {Nawaf Alhazmi};
  \node[above right=-3pt of 7] {Waleed Alshehri};
  \node[above left=-5pt and 1pt of 11] {Marwan Al-Shehhi};
  \node[left=2pt of 14] {Hamza Alghamdi};
  
  \path[draw,thick]
  (18) edge[e4c edge]  (19)
  (17) edge[e4c edge]  (19)
  (16) edge[e4c edge]  (18)
  (14) edge[e4c edge]  (19)
  (14) edge[e4c edge]  (18)
  (14) edge[e4c edge]  (17)
  (14) edge[e4c edge]  (15)
  (13) edge[e4c edge]  (14)
  (12) edge[e4c edge]  (15)
  (11) edge[e4c edge]  (16)
  (11) edge[e4c edge]  (12)
  (10) edge[e4c edge]  (11)
  (1) edge[e4c edge]  (16)
  (8) edge[e4c edge]  (9)
  (7) edge[e4c edge]  (10)
  (7) edge[e4c edge]  (9)
  (7) edge[e4c edge]  (8)
  (6) edge[e4c edge]  (16)
  (6) edge[e4c edge]  (11)
  (6) edge[e4c edge]  (10)
  (4) edge[e4c edge]  (16)
  (4) edge[e4c edge]  (11)
  (1) edge[e4c edge]  (2)
  (1) edge[e4c edge]  (5)
  (1) edge[e4c edge]  (6)
  (1) edge[e4c edge]  (11)
  (1) edge[e4c edge]  (3)
  (3) edge[e4c edge]  (4)
  (3) edge[e4c edge]  (5)
  (3) edge[e4c edge]  (6)
  (3) edge[e4c edge]  (14)
  (3) edge[e4c edge]  (17)
  (3) edge[e4c edge]  (19)
  ;
\end{tikzpicture}
\caption{September 11 terrorist attack network. Group Attachment Centrality identifies Hamza Alghamdi (\#14) and Marwan Al-Shehhi (\#11) as the pair of nodes whose removal would most disrupt network interactions. In turn, Group Degree Centrality selects Nawaf Alhazmi (\#3) and Marwan Al-Shehhi (\#11).}
\label{fig:wtc}
\end{figure}
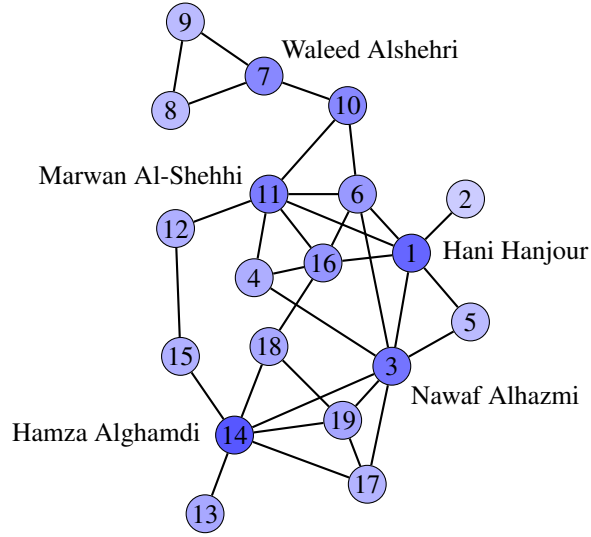

\begin{table}[t!]
\centering

\addtolength{\tabcolsep}{-3pt}    
\begin{tabular}{|c|l|c|c|}
\hline
\textbf{Rank} & \textbf{Name} & \textbf{Attachment} & \textbf{Degree} \\ \hline
1 & Hamza Alghamdi (\#14) & 3.32 & 6 \\ 
2 & Hani Hanjour (\#1) & 3.00 & 6 \\ 
3 & Nawaf Alhazmi (\#3) & 2.73 & 7 \\ 
4 & Marwan Al-Shehhi (\#11) & 2.59 & 6 \\ 
5 & Waleed Alshehri (\#7) & 2.33 & 3 \\ 
-- & Abdul Aziz Al-Omari (\#10) & 2.33 & 3 \\ 
7 & Ziad Jarrah (\#16) & 2.06 & 5 \\ 
8 & Mohamed Atta (\#6) & 2.02 & 5 \\ 
9 & Saeed Alghamdi (\#19) & 1.73 & 4 \\ 
10 & Ahmed Al-Haznawi (\#18) & 1.71 & 3 \\ 
11 & Salem Alhazmi (\#4) & 1.58 & 3 \\ 
12 & Mohand Alshehri (\#15) & 1.55 & 2 \\ 
-- & Fayez Ahmed (\#12) & 1.55 & 2 \\ 
14 & Ahmed Alnami (\#17) & 1.50 & 3 \\ 
15 & Khalid Al-Mihdhar (\#5) & 1.33 & 2 \\ 
-- & Wail Alshehri (\#8) & 1.33 & 2 \\ 
-- & Satam Suqami (\#9) & 1.33 & 2 \\ 
18 & Majed Moqed (\#2) & 1.00 & 1 \\ 
-- & Ahmed Alghamdi (\#13) & 1.00 & 1 \\ 
\hline
\end{tabular}
\caption{Centralities of nodes in the network from Figure~\ref{fig:wtc}, ranking according to Attachment Centrality.}
\label{tab:nodes_sorted_attachment}
\end{table}

\begin{example}\label{example:wtc}
Consider the graph in Figure~\ref{fig:wtc} that represents the terrorist group responsible for the September 11 attacks~\citep{Krebs:2002:mapping}.
While the top three nodes according to Attachment Centrality are \#14, \#1, and \#3, Group Attachment Centrality selects \#11 and \#14 as the best pair, with a score of $6.33$.
This is because removing both nodes detaches nodes \#12 and \#15 from the rest of the graph.

For Group Degree Centrality, the pair \#3 and \#11 is the unique group with a score of $7 + 6 = 13$. This follows from the fact that node \#3 is incident to both nodes \#14 and \#1.
\hfill $\lrcorner$
\end{example}

It is known that on trees both Attachment Centrality and Degree Centrality coincides~\citep{Skibski:etal:2019:attachment}.
We will use this fact in the following example.

\begin{example}\label{example:attachment}
Consider the graph in Figure~\ref{fig:attachment}, obtained from Figure~\ref{fig:two_paths} by adding the edges $\{A,C\}$ and $\{B,D\}$.
Due to isomorphism, all standard centrality measures assign equal importance to nodes $A$, $B$, $C$, and $D$.
Moreover, these nodes are more important than the others.

Now consider the problem of selecting two nodes.
According to Group Degree Centrality (Equation~\eqref{eq:group_degree}), a group's score is determined by the number of edges incident to it.
Hence, the groups $\{A,C\}$ and $\{D,E\}$ achieve the highest score, since removing them decreases the total number of edges by six.

In contrast, according to Attachment Centrality, nodes $A$-$D$ each have a centrality value of $2.5$, their immediate neighbors have centrality $2$, and the leaf nodes have centrality $1$.
Group Attachment Centrality selects the nodes whose removal most significantly disconnects the network.
If nodes $\{A,D\}$ (or $\{B,C\}$) are removed, the graph decomposes into four separate paths.
The centrality of this group equals $2.5 + 3.0 = 5.5$ (here, we apply the Vertex Exclusion principle: after removing $A$, the centrality of $D$ becomes $3$).

By contrast, removing group $\{A,B\}$ (or $\{C,D\}$) leaves one long path intact and creates two short paths.
As a result, this group performs worse, with a total score of $2.5 + 2.0 = 4.5$.
\hfill $\lrcorner$
\end{example}

Our axiomatizations are the first ones of Group Attachment and Degree Centralities.
They can be considered an extension of two existing joint axiomatizations of their node analogs~\citep{Skibski:etal:2019:attachment,Sosnowska:Skibski:2017:weighted}, but they are partially different.
In particular, other axiomatizations made use of locality axioms (under the name Locality and Additivity) and axioms that limit the value of centrality (such as Normalization and Monotonicity/Star-Max). 
In turn, we depend on duality of our normalization condition (Node/Edge-Normalization axioms).

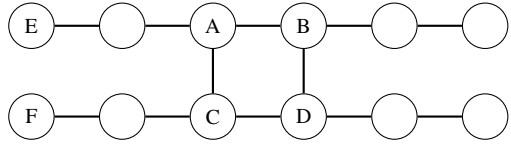
\begin{figure}[t]
\centering
\begin{tikzpicture}[x=6cm,y=6cm] 
  \tikzset{     
    e4c node/.style={circle,draw,minimum size=0.6cm,inner sep=0,font=\footnotesize}, 
    e4c edge/.style={sloped,above,font=\footnotesize}
  }
  \node[e4c node] (1) at (0.0, 0.2) {E}; 
  \node[e4c node] (2) at (0.2, 0.2) {}; 
  \node[e4c node] (3) at (0.4, 0.2) {A}; 
  \node[e4c node] (4) at (0.6, 0.2) {B}; 
  \node[e4c node] (5) at (0.8, 0.2) {}; 
  \node[e4c node] (6) at (1.0, 0.2) {}; 
  \node[e4c node] (7) at (0.0, 0.00) {F}; 
  \node[e4c node] (8) at (0.2, 0.00) {}; 
  \node[e4c node] (9) at (0.4, 0.00) {C}; 
  \node[e4c node] (10) at (0.6, 0.00) {D}; 
  \node[e4c node] (11) at (0.8, 0.00) {}; 
  \node[e4c node] (12) at (1.0, 0.00) {}; 

  \path[draw,thick]
  (1) edge[e4c edge]  (2)
  (2) edge[e4c edge]  (3)
  (3) edge[e4c edge]  (4)
  (4) edge[e4c edge]  (5)
  (5) edge[e4c edge]  (6)
  (7) edge[e4c edge]  (8)
  (8) edge[e4c edge]  (9)
  (9) edge[e4c edge]  (10)
  (10) edge[e4c edge]  (11)
  (11) edge[e4c edge]  (12)
  (3) edge[e4c edge]  (9)
  (4) edge[e4c edge]  (10)
  ;
\end{tikzpicture}
\caption{Graph for Example~\ref{example:attachment}. Both Group Degree and Attachment Centralities select $\{A,D\}$ as the best pair of nodes.}\label{fig:attachment}
\end{figure}

\section{Computational Results}

For group centrality measures, two algorithmic questions are natural.  
The first is to compute the centrality of a specific group.  
The second is to find a group---usually of a specific size---that maximizes the centrality.  
In this section, we present general results for both problems for the whole class of vitality indices, under minimal assumptions on the index itself.

Let $F$ be an arbitrary vitality index. Fix a graph $G = (V, E)$ and a group $S \subseteq V$.  
Using the Vertex Exclusion axiom again, we have that for every node $v \in S$, it holds:
\begin{equation}\label{eq:complexity}
F_S(G) = F_{\{v\}}(G) + F_{S \setminus \{v\}}(G - v).
\end{equation}
This implies a simple $|S|$-step recursive procedure for computing the centrality of $S$, based on the centrality of individual nodes, for all vitality indices.  
This leads to the following statement.

\begin{corollary}
If a node vitality index $F$ of a single node in a graph with $n$ nodes can be computed in time $p(n)$, then for a group vitality index $F'$ that extends $F$, the centrality of any set of vertices $S$ can be computed in time $|S| \cdot p(n)$.
\end{corollary}

\begin{proof}
Immediate from Equation~\eqref{eq:complexity}.
\end{proof}

Clearly, if computing a node vitality index is NP-hard, then computing the corresponding group vitality index is also NP-hard.

To analyze the second problem, we require some mild additional assumptions.  
We assume that the invariant function is \emph{anonymous}, i.e., it gives the same value for isomorphic graphs: if for two graphs $G = (V, E)$ and $G' = (V', E')$ there exists a bijection $b: V \rightarrow V'$ such that $E' = \{\{b(u), b(v)\} : \{u, v\} \in E\}$, then $f(G) = f(G')$.  
Moreover, it is natural to assume that the graph assessment is minimized when there are no edges in the graph.  
We are not aware of any vitality index that violates these assumptions.  
As we now show, for every group vitality index that satisfies these assumptions, finding a group of a given size that maximizes the centrality is NP-hard.

\begin{proposition}
Let $F$ be a group vitality index with an anonymous invariant function $f$ that satisfies $f((V, E)) > f((V, \emptyset))$ whenever $|E| > 0$.  
Then finding a set $S$ of size $k$ that maximizes the centrality $F_S(G)$ is NP-hard.
\end{proposition}

\begin{proof}
We use a reduction from the \textsc{$k$-Vertex Cover} problem.
Let $G$ be an arbitrary graph, and suppose we have an algorithm that finds a set $S$ of size $k$ that maximizes the centrality $F_S(G)$.  
Let $S^*$ be the returned set.  
If $S^*$ is a vertex cover of $G$, then the answer to the \textsc{Vertex Cover} problem is ``yes''.  
If $S^*$ is not a vertex cover, then a vertex cover of size $k$ does not exist, and the answer is ``no''.  

Indeed, suppose a vertex cover $S'$ of size $k$ exists. Then,
\[
F_{S'}(G) = f(G) - f(G - S') = f(G) - f((V \setminus S', \emptyset)),
\]
since $G - S'$ is a set of isolated vertices.

On the other hand, since $S^*$ is not a vertex cover, $G - S^* = (V \setminus S^*, E)$ has some edges, i.e., $|E| > 0$.  
Therefore,
\[
F_{S^*}(G) = f(G) - f((V \setminus S^*, E)).
\]
From our assumption, $f((V \setminus S^*, E)) > f((V \setminus S', \emptyset))$, which implies $F_{S'}(G) > F_{S^*}(G)$---a contradiction, since $S^*$ was assumed to be the group with the largest centrality.
\end{proof}

\subsection{Approximation Algorithm}

We have already shown that for every vitality index, computing the value of a group is as hard as computing the value for single nodes. 
However, computing Attachment Centrality is \#P-complete, with polynomial-time algorithms known only for a narrow class of chordal graphs~\citep{Skibski:etal:2019:attachment}.  
To address this, we present an approximation algorithm based on a novel sampling method for the Union Shapley value.  

To develop the approximation algorithm, we introduce a probabilistic interpretation of the Union Shapley value, which underpins group vitality indices.

\begin{proposition}\label{proposition:union_permutations}
The Union Shapley value of coalition $S$ in game $(N, g)$ satisfies:
\begin{equation}\label{eq:union_premutations}
US_S(N, g) = \frac{1}{|N|!} \sum_{\pi \in \Omega(N)} \sum_{i \in S} \frac{g(S_i^\pi) - g(S_i^\pi \setminus S)}{|S_i^\pi \cap S|},
\end{equation}
where $\Omega(N)$ is the set of all permutations of $N$, and $S_i^\pi$ is the set of elements that appear in $\pi$ no later than $i$, i.e., $S_i^\pi = \{j \in N : \pi(j) \le \pi(i)\}$.
\end{proposition}

\begin{proof}
For each coalition $T \subseteq N$, we count how many times $T$ appears as $S_i^\pi$, i.e., the number of pairs $(\pi, i)$ such that $\pi \in \Omega(N)$, $i \in S$, and $T = S_i^\pi$.  
This occurs if and only if:
\begin{itemize}
\item $T$ appears in the first $|T|$ positions of $\pi$, and
\item $i$ is the last element of $T$; thus $i \in T \cap S$.
\end{itemize}
For each $i \in T \cap S$, there are $(|T| - 1)! (|N| - |T|)!$ such permutations.  
Hence, the total number of such pairs is $|T \cap S| \cdot (|T| - 1)! (|N| - |T|)!$.
Thus, the right-hand side of Equation~\eqref{eq:union_premutations} can be reformulated as:
\begin{align*}
\frac{1}{|N|!} & \sum_{\pi \in \Omega(N)} \sum_{i \in S} \frac{g(S_i^\pi) - g(S_i^\pi \setminus S)}{|S_i^\pi \cap S|} = \\
& = \frac{1}{|N|!} \sum_{T\subseteq N} \sum_{\substack{\pi \in \Omega(N), i\in S : S_i^\pi = T}} \frac{g(T) - g(T \setminus S)}{|T \cap S|} \\
& = \frac{1}{|N|!} \sum_{T\subseteq N} (|T| \!-\! 1)!(|N| \!-\! |T|)! \cdot (g(T) \!-\! g(T\setminus S)) \\
& = US_S(N, g).
\end{align*}
which concludes the proof.
\end{proof}

\begin{algorithm}[!t]
\caption{Approximation of Group Attachment Centrality}
\label{algorithm:approx}
\begin{algorithmic}[1]

\REQUIRE A graph $G = (V,E)$ and a group $S \subseteq V$
\ENSURE The value $A_S(G)$

\STATE $\textit{result} = 0$
\FOR{$i=1 \ \TO\  \textit{number of samples}$}
	\STATE $\pi = $ random permutation of $V$
	\FORALL{$i \in S$}
		\STATE $\textit{result} \pluseq 2(|K(G[S_i^\pi \setminus S])| - |K(G[S_i^{\pi}])| + |S \cap S_i^\pi|)/|S_i^\pi \cap S|$
	\ENDFOR	
\ENDFOR
\RETURN $\textit{result}/\textit{number of samples}$

\end{algorithmic}
\end{algorithm}

Algorithm~\ref{algorithm:approx} presents the pseudocode for a Monte Carlo sampling algorithm to approximate Group Attachment Centrality. 
 
The algorithm utilizes Proposition~\ref{proposition:union_permutations} for the coalitional game defined by Attachment Centrality.  
In this game, the marginal contribution of $S$ to $T$ simplifies to:
\[
f(T) - f(T \setminus S) = 2 \big( |K(G[T \setminus S])| - |K(G[T])| + |S \cap T| \big).
\]

We note that this approximation algorithm can be easily extended to compute Attachment Centrality of all nodes.  
To do so, for each permutation, all players are considered in order, and their individual marginal contributions are computed and stored separately.  
A disjoint-set (union-find) data structure~\citep{Galler:Fisher:1964} can be used to efficiently track the connected components throughout the computation.

\section{Conclusions}

In this paper, we have performed a comprehensive analysis of group vitality indices.  
First, we proved that every vitality index admits a unique extension to a group vitality index.  
Second, we provided two axiomatic characterizations of the entire class and showed that Group Degree and Attachment Centralities are the unique indices satisfying specific normalization conditions.  
Furthermore, we demonstrated that computing a group vitality index for a given group is as hard as computing a node vitality index, and that identifying the best group is always NP-hard under mild assumptions.  
Finally, we designed an approximation algorithm for computing Group Attachment Centrality.

Our work represents an initial step toward understanding the axiomatic properties of group centrality measures, and many promising directions remain for future research.  
For instance, an axiomatic analysis of other classical definitions of group measures could be developed, potentially by extending existing axiomatizations of node centralities.  
In particular, axiomatizing a group variant of PageRank appears both important and exciting.  
Further exploration of the connections between group centrality measures and coalitional game theory is another promising avenue.  
Finally, the axiom of \emph{Vertex Exclusion} suggests a novel approach for defining group centrality, which may be applicable beyond vitality indices.

\section*{Ethical Statement}

There are no ethical issues.

\section*{Acknowledgments}

Natalia Kucharczuk was supported by the National Science Centre under Grant No. 2018/31/B/ST6/03201.
Oskar Skibski was supported by the National Science Centre under Grant No. 2023/50/E/ST6/00665.

%
\appendix 

\onecolumn

\section{Proof of \cref{lemma:union_shapley_form}}

\begin{proof}
We use the induction on $|V|$ and $|S|$.
If $V = \{v\}$, then we have $F_{\{v\}}(G) = \frac{0! \cdot 0!}{1!} F_{\{v\}}(G)$ which holds trivially.
Assume $|V| > 1$, take a node $v \in V$.
From Group Balanced Contributions applied to every node $u \in V \setminus \{v\}$ we have:
\[ |V| \cdot F_{\{v\}}(G) = \\ 
\sum_{u \in V \setminus \{v\}} F_{\{v\}}(G - u) + \sum_{u \in V} F_{\{u\}}(G) - \sum_{u \in V \setminus \{v\}} F_{\{u\}}(G - v).\]
Note that the last two sums boil down to $g^*(V)$ and $g^*(V \setminus \{v\})$.
Hence, using the inductive assumption for first sum we get:
\begin{align*}
F_{\{v\}}(G) = &
\frac{1}{|V|} \sum_{u \in V \setminus \{v\}} \sum_{T \subseteq V \setminus \{u\}} \frac{(|T|-1)!(|V|-|T|-1)!}{(|V|-1)!} \left(g^*(T) - g^*(T \setminus \{v\})\right)
+ \frac{1}{|V|} \left(g^*(V) - g^*(V \setminus \{v\})\right) \\
= & \frac{1}{|V|} \sum_{T \subseteq V, v \in T} (|V|-|T|) \frac{(|T|-1)!(|V|-|T|-1)!}{(|V|-1)!} \left(g^*(T) - g^*(T \setminus \{v\})\right) + \frac{1}{|V|} \left(g^*(V) - g^*(V \setminus \{v\})\right) \\
= & \sum_{T \subseteq V, v \in T} \frac{(|T|-1)!(|V|-|T|)!}{|V|!} \left(g^*(T) - g^*(T \setminus \{v\})\right) = \sum_{T \subseteq V} \frac{(|T|-1)!(|V|-|T|)!}{|V|!} \left(g^*(T) - g^*(T \setminus \{v\})\right).
\end{align*}
where the last equality comes from the fact that if a coalition $T$ does not contain $v$, then $g^*(T)=g^*(T \setminus \{v\})$.

Now, take any group $S \subseteq V$, $|S| > 1$ and a node $v \in V$.
From Group Balanced Contributions, we have:
\[ F_S(G) - F_{S \setminus \{v\}}(G - v) = F_{\{v\}}(G) - F_{\emptyset}(G - S). \]
Hence, we can use the inductive assumption again and get:
\begin{align*}
F_S(G) = & F_{S \setminus \{v\}}(G - v) + F_{\{v\}}(G) \\
= & \sum_{T \subseteq V, v \not \in T} \frac{(|T|-1)! (|V|-|T|-1)!}{(|V|-1)!} \left(g^*(T) - g^*(T \setminus (S \setminus \{v\}))\right)
+\\
& + \sum_{T \subseteq V, v \in T} \frac{(|T|-1)! (|V|-|T|)!}{|V|!} \left(g^*(T) - g^*(T \setminus \{v\})\right) \\
= & \sum_{T \subseteq V, v \not \in T} \left(\frac{(|T|-1)! (|V|-|T|)!}{|V|!}+\frac{(|T|)!(|V|-|T|-1)!}{|V|!}\right) \left(g^*(T) - g^*(T \setminus S)\right)\\
& + \sum_{T \subseteq V, v \in T} \frac{(|T|-1)! (|V|-|T|)!}{|V|!} \left(g^*(T) - g^*(T \setminus \{v\})\right) \\
= & \sum_{T \subseteq V, v \not \in T} \frac{(|T|-1)! (|V|-|T|)!}{|V|!} \left(g^*(T) - g^*(T \setminus S)\right)\\
& + \sum_{T \subseteq V, v \in T} \frac{(|T|-1)! (|V|-|T|)!}{|V|!} \left(g^*(T \setminus \{v\}) - g^*(T \setminus S)\right) \\
& + \sum_{T \subseteq V, v \in T} \frac{(|T|-1)! (|V|-|T|)!}{|V|!} \left(g^*(T) - g^*(T \setminus \{v\}))\right) \\
= & \sum_{T \subseteq V} \frac{(|T|-1)! (|V|-|T|)!}{|V|!} \left(g^*(T) - g^*(T \setminus S)\right).
\end{align*}
This concludes the proof.
\end{proof}

\newpage
\section{Proof of \cref{theorem:axioms}}

\begin{proof}
(a) $\Rightarrow$ (b):  
Let $F$ be a group vitality index with an invariant function $f$. Then:
\begin{align*}
F_S(G) &= f(G) - f(G - S) \\
&= f(G) - f(G - v) + f(G - v) - f(G - S) \\
&= F_{\{v\}}(G) + F_{S \setminus \{v\}}(G - v).
\end{align*}

(b) $\Rightarrow$ (c):  
Assume $F$ satisfies Vertex Exclusion.  
Take two coalitions $S, T \subseteq V$ and let $T = \{v_1,\dots,v_t\}$.  
From Vertex Exclusion we know that:
\[
F_T(G) = \sum_{i=1}^{t} F_{v_i}(G - \{v_1,\dots,v_{i-1}\}).
\]
Analogously, we get:
\begin{align}
F_{S \cup T}(G) &= \sum_{i=1}^{t} F_{v_i}(G - \{v_1,\dots,v_{i-1}\}) + F_{S \setminus T}(G - T) \notag \\
&= F_T(G) + F_{S \setminus T}(G - T). \label{eq:proof_axioms_1}
\end{align}
By considering coalition $S$ instead of $T$, we also have:
\begin{equation}\label{eq:proof_axioms_2}
F_{S \cup T}(G) = F_S(G) + F_{T \setminus S}(G - S).
\end{equation}
Combining Equations~\eqref{eq:proof_axioms_1} and \eqref{eq:proof_axioms_2} we get that $F$ satisfies Group Balanced Contributions.

(c) $\Rightarrow$ (a):  
Take a group centrality measure $F$ that satisfies Group Balanced Contributions.
We will show that if a centrality measure satisfies Balanced Contributions, then it is a group vitality index for the invariant function defined as follows:
\[
f(G) = \sum_{T \subseteq V} \left(\frac{(|T| - 1)! (|V| - |T|)!}{|V|!} \sum_{v \in V} F_v(G[T])\right) = \sum_{T \subseteq V} \left(\frac{(|T| - 1)! (|V| - |T|)!}{|V|!} g^*(T)\right).
\]
From Lemma~\ref{lemma:union_shapley_form} we know that
\[ F_S(G) = \sum_{T \subseteq V} \frac{(|T|-1)!(|V|-|T|)!}{|V|!} \left(g^*(T)-g^*(T \setminus S)\right)\]
holds for every graph $G = (V,E)$ and $S \subseteq V$.
It is enough to verify that:
\begin{align*} 
F_S(G) & = f(G) - f(G-S) \\
& = \sum_{T \subseteq V} \left(\frac{(|T| - 1)! (|V| - |T|)!}{|V|!} g^*(T)\right) - \sum_{T \subseteq V \setminus S} \left(\frac{(|T| - 1)! (|V| - |S|-|T|)!}{(|V|-|S|)!} g^*(T)\right) \\
& = \sum_{T \subseteq V} \left(\frac{(|T| - 1)! (|V| - |T|)!}{|V|!} g^*(T)\right) - \sum_{T \subseteq V \setminus S} g(T) \cdot \left(\frac{(|T| - 1)! (|V| - |S|-|T|)!}{(|V|-|S|)!} \right) \\
& = \sum_{T \subseteq V} \left(\frac{(|T| - 1)! (|V| - |T|)!}{|V|!} g^*(T)\right) - \sum_{T \subseteq V \setminus S} g(T) \cdot \left(\sum_{R \subseteq S} \frac{(|T|+|R| - 1)! (|V| - |R|-|T|)!}{|V|!}\right) \\
& = \sum_{T \subseteq V} \left(\frac{(|T| - 1)! (|V| - |T|)!}{|V|!} g^*(T)\right) - \sum_{T \subseteq V} \frac{(|T| - 1)! (|V| -|T|)!}{|V|!} g(T \setminus S) \\
& = \sum_{T \subseteq V} \frac{(|T| - 1)! (|V| - |T|)!}{|V|!} \left(g^*(T) - g^*(T \setminus S)\right)
\end{align*}
which concludes the proof.

\end{proof}

\newpage

\section{Proof of \cref{theorem:axioms_attachment}}

We will use the following lemma.
\begin{lemma}\label{lemma:star}
Let $F$ be an arbitrary vitality index satisfying Isolated Node and Leaf Node with constant $c$, then for any star $G^{\star} = (V,E)$ and vertex $v\in V$, we have $F_v(G^{\star}) = c \cdot D_v(G^{\star})$.
\end{lemma}
\begin{proof}
We will prove it by induction on the number of vertices.
If $V = \{v\}$, then from Isolated Node $F_v(G^{\star}) = 0 = D_v(G^{\star})$. 
    
Now, assume that $|V| > 1$ and the centrality of every node in the star with fewer than $|V|$ vertices is equal to its degree times $c$.
Let $v\in V$ be the center of the star $G^{\star}$, and $u\in V, u\not=v$ be a leaf. 
By the Leaf Node axiom, the centrality of any leaf is equal to its degree multiplied by $c$. Thus, it remains to prove that $F_v(G^{\star}) = c \cdot D_v(G^{\star})$.
    
Since any vitality index satisfies Balanced Contributions (as shown in Theorem~\ref{theorem:axioms}), we have:
\[ F_v(G^{\star}) - F_v(G^{\star}-u) = F_{u}(G^{\star}) - F_{u}(G^{\star}-v). \]
By the Isolated Node axiom, $F_{u}(G^{\star}-v) = 0$, and by the Leaf Node axiom, $F_{u}(G^{\star}) = c$. 
Moreover, from our inductive assumption, $F_v(G^{\star}-u) = c \cdot D_v(G^{\star}-u)$. Hence:
\[ F_v(G^{\star}) = c \cdot (D_v(G^{\star}-e) + 1) = c \cdot D_v(G^{\star}). \]
which completes the proof.
\end{proof}

\begin{proof}
Let us return to the proof of Theorem~\ref{theorem:axioms_attachment}.
Clearly, Group Attachment Centrality multiplied by any constant $c \in \mathbb{R}$ is a group vitality index that satisfies the axioms Isolated Node, Leaf Node (with constant $c$), and Node-Normalization.
Now, take an arbitrary vitality index that satisfies these three axioms.
We will prove that it is equal to Group Attachment Centrality multiplied by $c$.

Fix any graph $G = (V,E)$ and group $S \subseteq V$.
Let $C$ be any component of $G$, and consider the subgraph induced by $C$: $G[C] = (C, E[C])$.
From Node-Normalization, there exists a function $h: \mathbb{N} \rightarrow \mathbb{R}$ such that $\sum_{v\in C} F_v(G) = h(|C|)$.
    
Let $G^{\star} = (V^{\star}, E^{\star})$ be a star consisting of $|C|$ nodes. 
From Lemma~\ref{lemma:star}, for any $v\in V^{\star}$ we have 
\[ F_v(G^{\star}) = c \cdot D_v(G^{\star}), \] 
so that
\[ \sum_{v\in V^{\star}} F_v(G^{\star}) = c \cdot \sum_{v\in V^{\star}}D_v(G^{\star}) = 2c \cdot (|C|-1). \]

On the other hand, since $G^{\star}$ is connected, we also have: 
\[ \sum_{v\in V^{\star}} F_v(G^{\star}) = h(|V^{\star}|) = h(|C|).\] 
Hence, $h(C) = 2c \cdot (|C|-1)$.

Since this analysis holds for every connected component of the graph, it follows that:
\[ \sum_{v \in V} F_v(G) = \sum_{C \in K(G)} (2c \cdot (|C|-1)) = 2c \cdot (|V|-|K(G)|). \]
Moreover, this equality holds for any graph, including induced subgraphs of $G$.

Now, from Lemma~\ref{lemma:union_shapley_form} we have that:
\begin{align*} 
F_S(G) & = \sum_{T \subseteq V} \frac{(|T| - 1)! (|V| - |T|)!}{|V|!} \left(\sum_{v \in T} F_v(G[T]) - \sum_{v \in T \setminus S} F_v(G[T \setminus S])\right)\\
& = \sum_{T \subseteq V} \frac{(|T| - 1)! (|V| - |T|)!}{|V|!} \left(2c \cdot \left((|T|-|K(G[T])|) - (|T \setminus S|-|K(G[T \setminus S]|)\right)\right) \\
& = c \sum_{T \subseteq V} \frac{(|T| - 1)! (|V| - |T|)!}{|V|!} 2 (|T|-|K(G[T]|) - \\
& \ \ \ \ \  - c \sum_{T \subseteq V} \frac{(|T| - 1)! (|V| - |T|)!}{|V|!}  2 (|T \setminus S|-|K(G[T \setminus S]|) \\
& = c \cdot f^A(G) - 
c \cdot \sum_{T \subseteq V \setminus S} 2 (|T|-|K(G[T]|) \sum_{R \subseteq S}\frac{(|T|+|R| - 1)! (|V| - |T|-|R|)!}{|V|!}   \\
& = c \cdot f^A(G) - 
c \cdot \sum_{T \subseteq V \setminus S} 2 (|T|-|K(G[T]|) \frac{(|T| - 1)! (|V| - |S|-|T|)!}{(|V|-|S|)!} \\
& = c \cdot (f^A(G) - f^A(G - S)) = c \cdot A_S(G),
\end{align*}
which concludes the proof.
\end{proof}

\section{Proof of \cref{theorem:axioms_degree}}

\begin{proof}
Clearly, Group Degree Centrality multiplied by any constant $c \in \mathbb{R}$ is a group vitality index that satisfies Isolated Node, Leaf Node with constant $c$ and Edge-Normalization.
Take a vitality index that satisfies these axioms.
We will prove that it is equal to Group Degree Centrality multiplied by $c$.

Fix any graph $G = (V,E)$ and group $S \subseteq V$.
Let $C$ be any component of $G$, and consider the subgraph induced by $C$: $G[C] = (C, E[C])$.
From Edge Normalization, there exists a function $h: \mathbb{N} \rightarrow \mathbb{R}$ such that $\sum_{v\in C} F_v(G) = h(E[C])$.
    
Let $G^{\star} = (V^{\star}, E^{\star})$ be a star consisting of $|E[C]|$ edges. 
From Lemma~\ref{lemma:star}, for any $v\in V^{\star}$ we have $F_v(G^{\star}) = c \cdot D_v(G^{\star})$ so that 
\[ \sum_{v\in V^{\star}} F_v(G^{\star}) = c \cdot \sum_{v\in V^{\star}}D_v(G^{\star}) = 2c \cdot |E^{\star}| = 2c \cdot |E[C]|. \]
On the other hand, since $G^{\star}$ is connected, we also have:
\[ \sum_{v\in V^{\star}} F_v(G^{\star}) = h(|E^{\star}|) = h(|E[C]|). \] 
Hence, $h(E[C]) = 2c \cdot |E[C]|$.

Since this analysis holds for every component of the graph, it follows that:
\[ \sum_{v \in V} F_v(G) = \sum_{C \in K(G)} (2c \cdot |E[C]|) = 2 \cdot |E|. \]

Now, from Lemma~\ref{lemma:union_shapley_form} we have that:
\begin{align*} 
F_S(G) & = \sum_{T \subseteq V} \frac{(|T| - 1)! (|V| - |T|)!}{|V|!} \left(\sum_{v \in T} F_v(G[T]) - \sum_{v \in T \setminus S} F_v(G[T \setminus S])\right)\\
& = \sum_{T \subseteq V} \frac{(|T| - 1)! (|V| - |T|)!}{|V|!} \left(2c \cdot \left(|E[T]| - |E[T \setminus S]|\right)\right) \\
& = \sum_{\{u,v\} \in E \setminus E[V \setminus S]} 2c \cdot \sum_{T \subseteq V : u,v \in T} \frac{(|T| - 1)! (|V| - |T|)!}{|V|!} = c \cdot \left(|E| - |E[V \setminus S]|\right) = c \cdot D_S(G),
\end{align*}
which concludes the proof.
\end{proof}

\end{document}